\documentclass{article}
\usepackage{amsmath}
\usepackage{amssymb}
\usepackage{amsthm}
\usepackage{esint}
\usepackage{braket}
\usepackage{graphicx}
\usepackage{float}
\usepackage{booktabs}
\usepackage[title]{appendix}

\newtheorem{theorem}{Theorem}

\DeclareMathOperator{\real}{Re}

\DeclareMathOperator{\rect}{rect}
\DeclareMathOperator{\sinc}{sinc}

\title{Understanding the physics of coherent LiDAR}
\author{Alexander Y. Piggott}
\date{\today}

\begin{document}
\maketitle

%Coherent LiDAR (Light Detecting And Ranging) is a promising 3D imaging technology that provides significant advantages over more traditional LiDAR systems. In addition to being immune to ambient light, coherent LiDAR directly measures the velocity of moving objects by sensing Doppler shift of light, and can achieve exceptional depth accuracies.  The goal of this manuscript is to explain the basic physics of coherent LiDAR with rigorous derivations from first principles. We first discuss the sensitivity of coherent LiDAR, and derive the number of photons needed to robustly detect a target. We then turn our attention to the collection efficiency of coherent LiDAR, and show that signal strength is strongly dependent upon how well the laser beams are focused.

\section{Introduction}
Coherent LiDAR (Light Detecting And Ranging) is a promising 3D imaging technology that provides significant advantages over more traditional LiDAR systems. In addition to being immune to ambient light \cite{bbehroozpour_ieeecm2017}, coherent LiDAR directly measures the velocity of moving objects by sensing Doppler shift of light \cite{bbehroozpour_ieeecm2017, hdgriffiths_ecej1990}, and can achieve exceptional depth accuracies. Due to these attributes, it is becoming popular for autonomous vehicle applications. Coherent LiDAR is also the ranging technique of choice for a variety of groups attempting to develop fully solid-state LiDAR systems \cite{cvpoulton_ol2017, samiller_cleo2018, cvpoulton_ieeeqe2019, jriemensberger_nature2020}.

The goal of this manuscript is to explain the basic physics of coherent LiDAR with rigorous derivations from first principles. In particular, we strive to answer two key questions about coherent LiDAR:
\begin{enumerate}
\item What is the sensitivity of coherent LiDAR, i.e. how many photons need to be collected to detect a target?
\item What is the collection efficiency of coherent LiDAR, i.e. how many photons are actually returned from a target?
\end{enumerate}

\clearpage
\section{Architectures}

\begin{figure}[b!]
\centering
\includegraphics[scale=0.9]{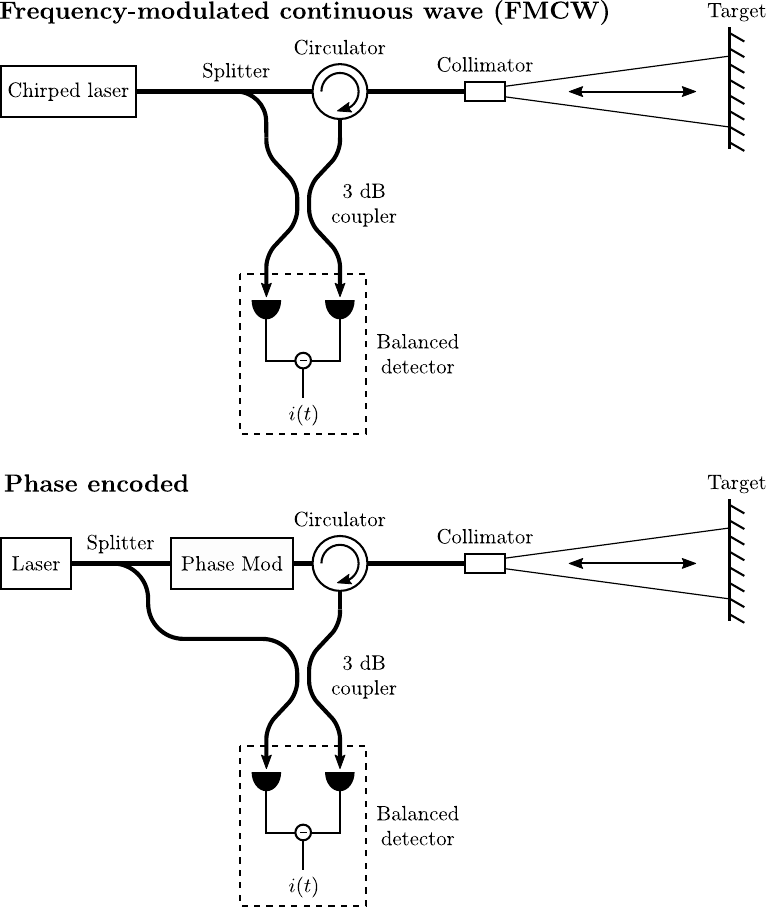}
\caption{The two basic coherent LiDAR architectures. \label{fig:architectures}}
\end{figure}

Amplitude modulated LiDAR systems use \emph{direct detection}, where scattered light from a target is directly detected by a photodiode or other photodetector. Although direct detection is simple to implement, system performance is typically limited by detector noise and ambient light.

In contrast, coherent LiDAR systems use \emph{coherent detection}, where the weak scattered signal is mixed with a strong local oscillator signal on a detector. The signal is significantly amplified by this mixing process, allowing relatively noisy detectors to be used. Due to the use of coherent detection, a coherent LiDAR system can directly measure the velocity of moving objects by sensing the Doppler shift of light. Furthermore, since a coherent detector can only detect light that is close in frequency to the local oscillator light, coherent LiDAR is virtually immune to interferene from ambient light.

As illustrated in figure \ref{fig:architectures}, coherent LiDAR systems typically use one of two schemes: the frequency-modulated continuous-wave (FMCW) scheme, and the phase-encoded scheme. The physical systems are almost identical to each other; the only difference is in the modulation scheme. In an FMCW system, a linearly chirped laser is used as both the transitted signal and local oscillator. Meanwhile, in a phase-encoded system, the transmitted signal is modulated in phase, while the un-modulated laser is used as the local oscillator. For simplicity, the LiDAR systems shown in figure \ref{fig:architectures} use balanced detectors. Although more sophisticated implementations may replace the balanced detector with a full coherent receiver, and introduce additional components to sample multiple poplarizations, this does not have a major influence on the basic physics of operation.

\subsection{Frequency-modulated continuous-wave (FMCW)}
The FMCW scheme uses a linearly chirped laser for both the transmitter and local oscillator \cite{bbehroozpour_ieeecm2017, hdgriffiths_ecej1990}, as illustrated in figure \ref{fig:fmcw_scheme}. In the case of a static target, the received signal is a simply a time-delayed version of the transmitted signal. By mixing the transmitted and received signals on an optical heterodyne receiver, the frequency difference between the transmitted and received signal can be extracted. The frequency difference is proportional to the round-trip travel time, and is therefore a measure of the target range.

If the target is moving, the received signal will have an additional frequency shift $\Delta f_\mathit{Doppler}$ proportional to the velocity $v$ due to the Doppler effect:
\begin{align}
\Delta f_\mathit{Doppler} \approx \frac{2 v}{\lambda},
\end{align}
Here, $\lambda$ is the wavelength of light.

To measure both the range and velocity to a target, triangular modulation is commonly used, where an up-chirp is immediately followed by a down-chirp. The measured frequency differences $\Delta f_\mathit{up}$ and $\Delta f_\mathit{down}$ during the up- and down-chirps respectively can then be used to compute the range $d$,
\begin{align}
d = \frac{c \left(\Delta f_\mathit{up} + \Delta f_\mathit{down}\right)}{4 r}
\end{align}
and the velocity $v$,
\begin{align}
v = \frac{\lambda \left(\Delta f_\mathit{up} - \Delta f_\mathit{down}\right)}{4}.
\end{align}
Here, $c$ is the speed of light, and $r$ is the chirp ramp rate.

\begin{figure}[h!tb]
\centering
\includegraphics[scale=0.8]{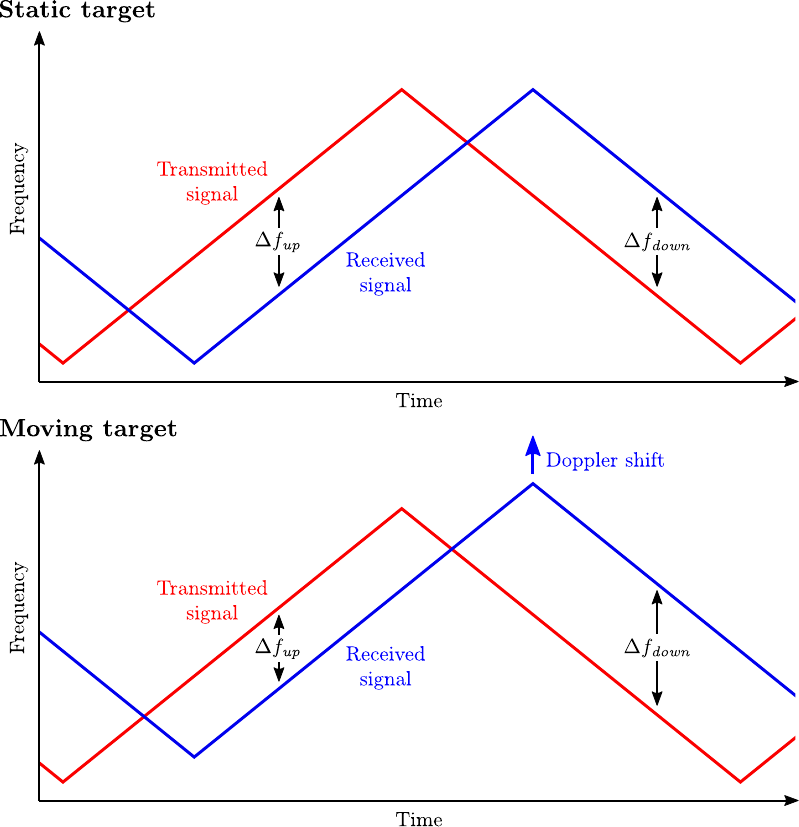}
\caption{Frequency-modulated continuous-wave LiDAR scheme. By following an up-chirp with a down-chirp, it is possible to extract both the range and velocity of a moving object. \label{fig:fmcw_scheme}}
\end{figure}

\subsection{Phase-encoded}
In the phase-encoded scheme, the transmitted beam is phase-modulated with a pseudo-random sequence of bits \cite{blackmore_patent2019_phase}. The received signal is detected on a coherent detector, and then digitally cross-correlated with the transmitted bit-stream to determine the time delay and therefore range to the target.

Although this approach is simpler to implement than the FMCW approach which requires a highly linear frequency chirp, there are a several challenges that must be overcome to create a successful system. First, the Doppler shifts can be very large, reaching up to $50 - 100~\mathrm{MHz}$ for automotive LiDARs, severely interfering with the cross-correlation step. Careful selection of pseudo-random sequences and additional digital signal processing is needed to handle Doppler shifted returns \cite{blackmore_patent2019_phase}. Furthermore, backreflections can elevate the noise floor, and need to be digitally subtracted before taking any cross correlations.

\clearpage
\section{Sensitivity of coherent detectors}
The sensitivity of amplitude-modulated LiDARs is typically limited by system imperfections such as detector dark current, amplifier noise, and ambient light. Although single-photon sensitive systems are permitted by the laws of physics, no commercially available system approaches this level of performance.

In contrast, it is quite straightforward to create a coherent LiDAR system that operates at the fundamental quantum limit of noise. This is achieved when the shot noise of the local oscillator light dominates over all other noise sources in the system, and is typically referred to as the \emph{shot noise limited regime} \cite{mjcollett_jmo1987, marubin_ol2007}. Shot-noise limited detection can be achieved even with relatively noisy photodetectors, as long as the local oscillator light is powerful enough to overcome any detector noise.

A key point to be made, however, is that the fundamental quantum limit for noise is substantially \emph{higher} for coherent LiDAR systems than amplitude-modulated systems. Amplitude-modulated LiDAR is fundamentally limited by thermal noise. At the optical frequencies used by LiDAR systems, thermal noise is essentially zero, and it is possible to robustly measure range by detecting individual photons (for example, with superconducting nanowire detectors \cite{kmnatarajan_sst2012}). In contrast, coherent LiDAR systems are limited by the shot noise of the local oscillator light, and require on the order of a few dozen photons for robust detection.

In this section, we will derive the single-to-noise ratio of coherent LiDAR systems as a function of the number of collected photons. For convenience, we will focus on the FMCW scheme, although phase-encoded schemes will have very similar results.

\subsection{Balanced optical heterodyne receiver}
Consider the ideal balanced optical heterodyne receiver shown in figure \ref{fig:balanced_det}.

\begin{figure}[h!tb]
\centering
\includegraphics[scale=0.9]{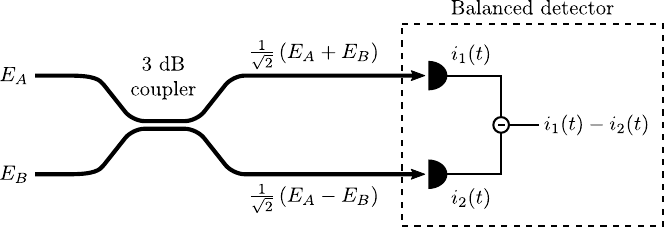}
\caption{Balanced optical heterodyne receiver. \label{fig:balanced_det}}
\end{figure}

The receiver combines light from the two inputs $A$ and $B$ on a 3 dB coupler, which is then detected by two photodiodes. The difference in photocurrent between these two photodiodes is given as the output of the receiver.

We will consider the simple case where the two inputs are single-frequency signals at slightly different frequencies. The input optical fields $E_A$ and $E_B$ would then be
\begin{align}
E_A &= A \cos \omega_A t \\
E_B &= B \cos \omega_B t,
\end{align}
with amplitudes $A$ and $B$ and optical frequencies $\omega_A$ and $\omega_B$ respectively. For convenience, we will also define the time-averaged photocurrents of the two inputs
\begin{align}
i_A 
  &= R\braket{E_A^2} = R A^2 \braket{\cos^2 \omega_A t} = \frac{R A^2}{2} \\
i_B
  &= R\braket{E_B^2} = R B^2 \braket{\cos^2 \omega_B t} = \frac{R B^2}{2},
\end{align}
where $R$ is the responsivity of the photodiodes in the balanced detector, and $\braket{\cdot}$ denotes the average value.

Our goal is to compute the photocurrents generated by the two photodiodes in the balanced detector. The electric field $E_1$ impinging on the first photodiode is given by
\begin{align}
E_1 
= \frac{1}{\sqrt{2}}(E_A + E_B) 
= \frac{1}{\sqrt{2}}(A \cos \omega_A t + B \cos \omega_B t).
\end{align}
The instantaneous power $E_1^2$ is then
\begin{align}
E_1^2 
= \frac{1}{2} \left(
    A^2 \cos^2 \omega_A t
    + B^2 \cos^2 \omega_B t
    + 2 A B \cos \omega_A t \, \cos \omega_B t
   \right).
\end{align}
Making use of the identities
\begin{align}
\cos^2 \theta 
  &= \frac{1}{2} + \frac{1}{2} \cos 2 \theta \\
\cos\theta \, \cos\phi 
  &= \frac{1}{2} \cos(\theta + \phi) 
  + \frac{1}{2} \cos(\theta - \phi),
\end{align}
our expression for the instantaneous power becomes
\begin{align}
E_1^2 = \frac{1}{4}
    \Big(&
        A^2 + B^2
        + A^2 \cos 2 \omega_A t
        + B^2 \cos 2 \omega_B t \\
    &   + 2 A B \cos(\omega_A + \omega_B) t
        + 2 A B \cos(\omega_A - \omega_B) t
    \Big).
\end{align}
The photocurrent $i_1(t)$ is related to the instantaneous power $E_1^2$ by the responsivity $R$, and is given by
\begin{align}
i_1(t) = \frac{R}{4} \left(A^2 + B^2 + 2 A B \cos(\omega_A - \omega_B) \right).
\label{eqn:photocurrent_i1_0}
\end{align}
Here, we have dropped all high-frequency terms since the photodiode can at best respond to frequencies in the gigahertz range. Rewriting this expression in terms of the input photocurrents $i_A$ and $i_B$ yields
\begin{align}
i_1(t) = \frac{1}{2}\left(i_A + i_B\right) 
    + \sqrt{i_A \, i_B} \cos\left(\omega_A - \omega_B\right) t.
\end{align}

Similarly, the field $E_2$ impinging on the second photodiode is
\begin{align}
E_2
= \frac{1}{\sqrt{2}}(E_A - E_B) 
= \frac{1}{\sqrt{2}}(A \cos \omega_A t - B \cos \omega_B t),
\end{align}
and the corresponding photocurrent is
\begin{align}
i_2(t) = \frac{1}{2}\left(i_A + i_B\right) 
    - \sqrt{i_A \, i_B} \cos\left(\omega_A - \omega_B\right) t.
\end{align}

The output current $i(t)$ is the difference between the individual photocurrents $i_1(t)$ and $i_2(t)$, which will therefore be
\begin{align}
i(t) = i_1(t) - i_2(t) = 2 \sqrt{i_A \, i_B} \cos(\omega_A - \omega_B)t.
\label{eqn:baldet_output}
\end{align}
Thus, an optical heterodyne detector directly measures the \emph{frequency difference} $\omega_A - \omega_B$ between the two input beams. The FMCW scheme makes direct use of this to measure the range to targets.

\subsection{Shot-noise limited detection}
In a coherent LiDAR system, an optical heterodyne receiver is used to detect photons scattered from the target. Local oscillator light is fed into input $A$ of the heterodyne receiver, and the scattered light is fed into input $B$. The scattered signal is typically extremely weak, and may consist of as few as a couple dozen photons. The local oscillator photocurrent $R\braket{E_A^2} = i_{LO}$ will thus typically be orders of magnitude stronger than the signal photocurrent $R\braket{E_B^2} = i_{sig}$.

Ideally, the optical heterodyne receiver is operated in the shot-noise limited regime, where local oscillator shot noise dominates over all other noise sources. In this regime, the signal-to-noise ratio depends only upon the number of collected photons, as we will show below.

\subsubsection{FMCW signal}
If we ignore all noise sources, we can see from equation \ref{eqn:baldet_output} that the idealized balanced detector output in an FMCW system will be
\begin{align}
i(t) = 2 \sqrt{i_{LO} \, i_{sig}} \cos(\omega_A - \omega_B)t,
\label{eqn:det_out}
\end{align}
where the difference frequency $\Delta \omega = \omega_A - \omega_B$ is proportional to the range to the target. This FMCW signal can be extremely weak, and is typically detected in the frequency domain by digitizing the detector output and taking a fast Fourier transform (FFT).

Our goal is to determine the signal-to-noise ratio of this signal as a function of the number of collected photoelectrons, which in turn depends upon the signal photocurrent $i_{sig}$ and integration time $T$. We can model the finite integration time by applying a window function $w(t)$ to the detector output $i(t)$ to obtain the windowed detector output $i_w(t)$,
\begin{align}
i_w(t) = i(t) \, w(t).
\end{align}
Next, to compute the signal-to-noise ratio in the frequency domain, we need to calculate the Fourier transform $I_w(f)$ of the windowed detector output. The convolution theorem states that
\begin{align}
I_w(f) = I(f) \ast W(f),
\label{eqn:det_conv_thm}
\end{align}
where $I(f)$ and $W(f)$ are the Fourier transforms of $i(t)$ and $w(t)$ respectively, and $\ast$ is the convolution operator
\begin{align}
(f \ast g) (x) = \int_{-\infty}^{\infty} f(y) \, g(x - y) \, dy.
\end{align}

We first turn our attention to computing the detector photocurrent term $I(f)$, which is given in the time domain by equation \ref{eqn:det_out}. Since the Fourier transform of $\cos \omega t$ is
\begin{align}
\int_{-\infty}^{\infty} \cos \omega t \: e^{-2 \pi j f t} \: dt
    = \frac{1}{2} \left[
        \delta\left( f - \frac{\omega}{2 \pi} \right)
        + \delta\left( f + \frac{\omega}{2 \pi} \right)
    \right],
\end{align}
the Fourier transform of $i(t)$ is given by
\begin{align}
I(f)
    = \sqrt{i_{LO} \, i_{sig}} 
    \Big[    
        \delta\left( f - f_A + f_B \right)
        + \delta\left( f + f_A - f_B \right)
    \Big].
\label{eqn:ft_det_output}
\end{align}
Here, we have defined $f_A = \omega_A / 2 \pi$ and $f_B = \omega_B / 2 \pi$.

Next, we turn our attention to computing the Fourier transform $W(f)$ of the window function $w(t)$. In any practical system, the window function is carefully be chosen to optimize system performance. Due to the presence of unavoidable backreflections within the LiDAR system, window functions that minimize sidelobes such as the Hann window are generally preferred. For the purposes of our current discussion, however, we will assume that we have a rectangular window function as shown in figure \ref{fig:rect_window}, which maximizes the signal-to-noise ratio. 
\begin{figure}[h!tb]
\centering
\includegraphics[scale=0.9]{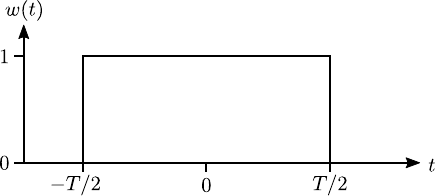}
\caption{Rectangular window applied to the FMCW signal. \label{fig:rect_window}}
\end{figure}
In terms of the rectangular function
\begin{align}
\rect(t) =
    \begin{cases}
        1, & |t| < \frac{1}{2}\\
        0, & \mathrm{otherwise},
    \end{cases} 
\end{align}
the rectangular window $w(t)$ can be written as
\begin{align}
w(t) = \rect\left( \frac{t}{T} \right).
\end{align}
Since the Fourier transform of $\rect t$ is
\begin{align}
\int_{-\infty}^{\infty} \rect t \: e^{-2 \pi j f t} \: dt
    = \frac{\sin \pi f}{\pi f}
    = \sinc f,
\end{align}
the Fourier transform $W(f)$ of the window function is therefore
\begin{align}
W(f) = T \sinc T f.
\label{eqn:ft_window}
\end{align}

Combining equations \ref{eqn:det_conv_thm}, \ref{eqn:ft_det_output}, and \ref{eqn:ft_window}, we find that that the Fourier transform of the windowed detector function $I_w(f)$ is 
\begin{align}
I_w(f) 
    = T \sqrt{i_{LO} \, i_{sig}}
    \Big[
        \sinc T \left( f - f_A + f_B \right)
	+ \sinc T \left( f + f_A - f_B \right)
    \Big].
\end{align}
The two-sided energy spectral density is therefore
\begin{align}
|I_w(f) |^2
    = T^2 \, i_{LO} \, i_{sig}
    \left[
        \sinc^2 T \left( f - f_A + f_B \right)
	+ \sinc^2 T \left( f + f_A - f_B \right)
    \right]
\label{eqn:esd_det_output_window}
\end{align}
where we have ignored any cross-terms by assuming that $|f_A - f_B| \gg 1/T$.

\subsubsection{Shot noise}
Now that we have an understanding of the signal strength in an FMCW LiDAR system, our next step is to determine the noise floor set by the photocurrent shot noise. The shot noise depends upon the \emph{sum} $i_1(t) + i_2(t)$ of the photocurrents generated by the photodiodes, which is approximately equal to the local oscillator photocurrent $i_{LO}$ in a typical LiDAR system where $i_{LO} \gg i_{sig}$. The two-sided power spectral density of the shot noise is therefore
\begin{align}
S(f) = q \, i_{LO},
\label{eqn:psd_shotnoise}
\end{align}
where $q$ is the charge of an electron.

To compute the signal-to-noise ratio, we need to find the expected energy spectral density $\braket{|I_n(f)|^2}$ of the windowed shot-noise. This is given by the convolution of the noise power spectral density $S(f)$ with the energy spectral density $|W(f)|^2$ of the windowing function:
\begin{align}
\braket{|I_n(f)|^2} = S(f) \ast |W(f)|^2.
\label{eqn:esd_shotnoise_0}
\end{align}
For a detailed derivation of this result, please see Appendix \ref{appendix:esd_windowed_random}.

Using our expressions \ref{eqn:ft_window} and \ref{eqn:psd_shotnoise} and for $W(f)$ and $S(f)$ respectively, we find that the energy spectral density is
\begin{align}
\braket{|I_n(f)|^2}
    &= \int_{-\infty}^{\infty} S(f - f') \, |W(f')|^2 \, df' \nonumber \\
    &= \int_{-\infty}^{\infty} q \, i_{LO} \,  T^2 \, \sinc^2 T f' \, df' \nonumber \\
    &= q \, i_{LO} \, T,
\label{eqn:esd_shotnoise_1}
\end{align}
where we have made use of
\begin{align}
\int_{-\infty}^{\infty} \sinc^2 x \, dx = 1.
\end{align}

\subsubsection{Signal-to-noise ratio}
Since the FMCW signal is uncorrelated with the shot noise, the total expected energy spectral density $\braket{|I_t(f)|^2}$ is simply the sum of the FMCW signal $|I_w(f)|^2$ and the shot noise $\braket{|I_n(f)|^2}$,
\begin{align}
\braket{|I_t(f)|^2} = |I_w(f)|^2 + \braket{|I_n(f)|^2},
\end{align}
which we have plotted below in figure \ref{fig:sinc_psd}.

\begin{figure}[h!tb]
\centering
\includegraphics[scale=0.9]{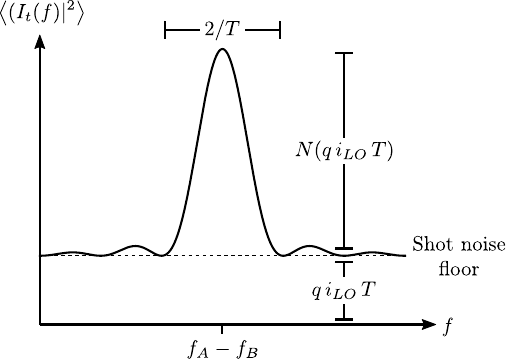}
\caption{Expected energy spectral density for an FMCW LiDAR signal. There is a strong peak at the difference frequency $f_A - f_B$, corresponding to the range to the target. The signal-to-noise ratio is exactly equal to the number of signal photoelectrons $N$. \label{fig:sinc_psd}}
\end{figure}

To find the signal-to-noise ratio, we take the ratio of the FMCW signal and shot noise spectral densities. Using equations \ref{eqn:esd_det_output_window} and \ref{eqn:esd_shotnoise_1}, this yields
\begin{align}
\frac{|I_w(f)|^2}{\braket{|I_n(f)|^2}}
    = N \left[
        \sinc^2 T \left( f - f_A + f_B \right)
	+ \sinc^2 T \left( f + f_A - f_B \right)
    \right],
\end{align}
where $N$ is the number of signal photoelectrons collected during the acquisition time $T$,
\begin{align}
N = \frac{i_{sig} \, T}{q}.
\end{align}
Since $\sinc^2 x \leq 1$, the signal-to-noise ratio $\mathit{SNR}$ in the frequency domain is exactly equal to the number of signal photoelectrons $N$:
\begin{align}
\mathit{SNR} = N.
\end{align}
This is the basis of the common claim that coherent LiDAR systems have ``single-photon sensitivity''. However, a signal-to-noise ratio of one is not nearly high enough for robust detection, and in practice coherent LiDAR systems require several dozen photons for a detection.

%\subsection{Detection probability and false return rate}
%Most obvious thing to do is set a detection threshold in PSD
%Any real signals above detection threshold are successfully detected
%Any noise signals above detection threshold are false returns

\clearpage
\section{Collection efficiency of coherent LiDAR}
Coherent LiDAR systems have a very fundamental and key limitation: \emph{they can only detect light in a single optical mode}. This is immediately obvious if we consider a coherent LiDAR system constructed using single-mode optical fibers: only scattered light that makes its way back into the guided mode of the collecting fiber can be detected. However, this also turns out to be true no matter how the coherent LiDAR is constructed\cite{aesiegman_ao1966}, even if only free-space optical components such as lenses and beamsplitters are used! 

This obviously has severe implications for the collection efficiency of a coherent LiDAR system. In an amplitude modulated LiDAR system that uses direct detection, we can trivially increase the collection efficiency by increasing the size of the detector aperture. However, in a coherent LiDAR system, increasing the size of the detector aperture does not necessarily improve the collection efficiency \cite{jywang_ao1982}.

As we will show later, the average optical power $\braket{P}$ collected by an ideal coherent LiDAR from a diffuse target is
\begin{align}
\braket{P} = \frac{\lambda^2 I}{2 \pi}.
\label{eqn:coherent_col_unif}
\end{align}
Here, $I$ is the intensity of the scattered light at the surface of the target, and $\lambda$ is the wavelength of light. Since the wavelength of light is only $\sim 10^{-6}~\mathrm{m}$, only a tiny fraction of the scattered light is collected by the LiDAR system. A particularly striking feature of equation \ref{eqn:coherent_col_unif} is that the optical power depends \emph{only} upon the scattered light intensity $I$, and therefore the intensity of the illuminating beam.

Many coherent LiDARS employ a \emph{monostatic} configuration, where the transmitted and collected light share the same optical path. In the case of an ideal monostatic LiDAR, the average optical power $\braket{P}$ becomes
\begin{align}
\braket{P} = \frac{\lambda^2 P_{s}}{2 \pi A_\mathit{eff}}
\label{eqn:coherent_col_nonunif}
\end{align}
where $P_{s}$ is the total power of the scattered light, and $A_\mathit{eff}$ is the effective area of the beam. The signal strength is inversely proportional to the beam area, indicating that tightly focusing the transmit and collection beams is critical to achieving high performance in a coherent LiDAR system.

The effective beam area of several common beam shapes is listed below.
\begin{center}
\begin{tabular}{ c c c }
  \toprule
  Type & Intensity $I(x, y)$ & Effective area ($A_\mathit{eff}$) \\
  \midrule
  Flat-top
    & $\begin{cases}
        1, & x^2 + y^2 < r^2 \\
        0, & \mathrm{otherwise.}
      \end{cases}$
    & $\pi r^2$ \vspace{0.2cm} \\ 

   Gaussian 
    & $\exp\left( - \dfrac{2 (x^2 + y^2)}{w^2} \right)$
    & $\pi w^2$ \vspace{0.2cm} \\
  \bottomrule
\end{tabular}
\end{center}

%Scattering off rough surface, scatter into all available modes with roughly equal probability
%Count number of modes: area / lambda**2

\subsection{Setup}
Since most modern coherent LiDARs are based on fiber-optics or other integrated photonic waveguides, we will study the couping efficiency of scattered light into a single-mode waveguide.

\begin{figure}[h!tb]
\centering
\includegraphics[scale=0.9]{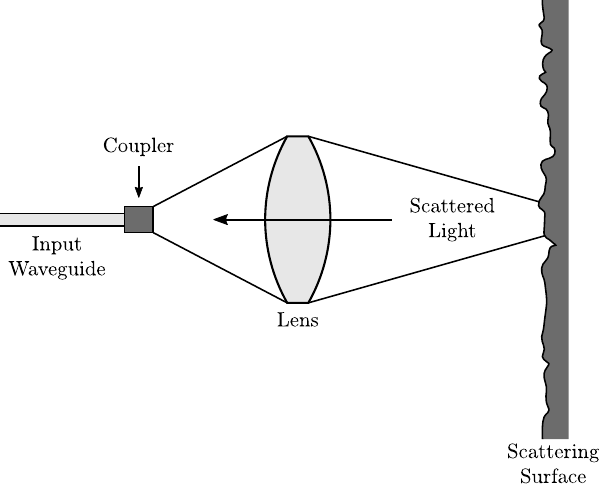}
\caption{Collection optics for a typical coherent LiDAR. Scattered light is focused using a lens, and coupled into a single-mode waveguide via a coupler. This coupler may be a grating coupler in the case of a integrated photonic waveguide, or possibly just a cleaved optical fiber in the case of a fiber-optic based system. \label{fig:waveguide_setup}}
\end{figure}

Our goal is to determine the expected power coupled into the input waveguide, which will ultimately allow for an understanding of the signal strength and thus the useful range of a coherent LIDAR. In particular, we will focus on the case of diffuse scattering surfaces.

\subsection{Power in input waveguide}
Our first step is to express the recieved power in terms of a field overlap integral. Consider a closed surface $S$ that encapsulates the lens and waveguide coupler:

\begin{figure}[h!tb]
\centering
\includegraphics[scale=0.9]{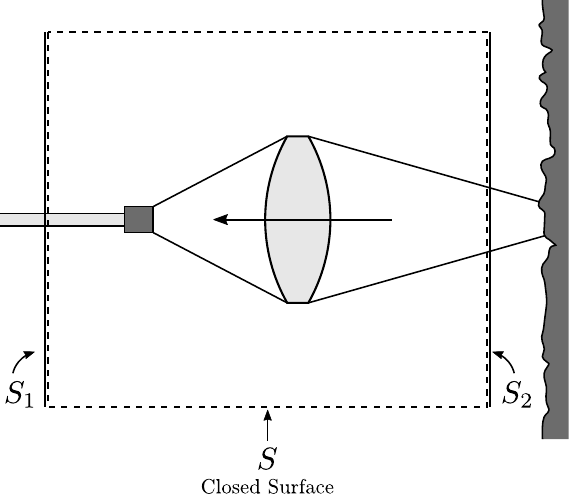}
\caption{To compute the power coupled into the waveguide, we will integrate the fields over surface $S$. The surface is chosen such that only the front and back faces $S_1$ and $S_2$ have non-negligible fields.  \label{fig:waveguide_recip}}
\end{figure}

\noindent
The surface $S$ is chosen so that only the front and back faces $S_1$ and $S_2$ have non-negligible fields. 

We will denote the \emph{scattered} electric and magnetic fields as $\mathbf{E}$ and $\mathbf{H}$. We will also make use of the \emph{modal} fields $\mathbf{E}_m$ and $\mathbf{H}_m$, which are the fields produced by injecting the desired waveguide mode \emph{back} into the system.

\begin{figure}[h!tb]
\centering
\includegraphics[scale=0.9]{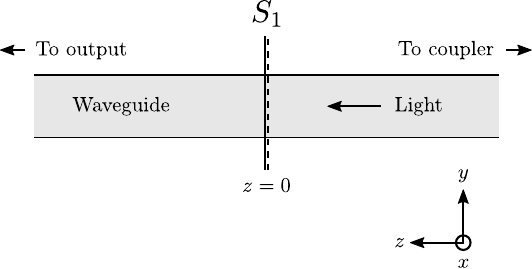}
\caption{The power in the desired waveguide mode can be found using a mode overlap integral. The scattered light propagates through the waveguide in the $+z$ direction, and the surface $S_1$ is located at $z = 0$. \label{fig:waveguide_overlap}}
\end{figure}
To find the power in the waveguide mode, we can perform a mode overlap integral over surface $S_1$, as illustrated in figure \ref{fig:waveguide_overlap}. We choose our local coordinate system such that $S_1$ is located in the $z = 0$ plane, and the scattered light propagates in the $+z$ direction.

The modal fields $\mathbf{E}_m$ and $\mathbf{H}_m$ in the vicinity of surface $S_1$ are simply the fields of the waveguide mode propagating in the $-z$ direction (i.e. \emph{away} from the output). These have the form
\begin{align}
\mathbf{E}_m(x,y,z) &= \mathcal{E}_m(x,y) e^{j \beta z} \\
\mathbf{H}_m(x,y,z) &= \mathcal{H}_m(x,y) e^{j \beta z},s
\end{align}
where $\beta$ is the propagation constant of the waveguide.

We will assume that the waveguide is reciprocal, lossless, and has reflection symmetry across some plane perpendicular to $\hat{z}$. The transverse components of the modal fields $\mathcal{E}_m$ and $\mathcal{H}_m$ can then be taken to be purely real (see Appendix \ref{appendix:wg_modes}). For convenience, we further assume that the modal fields are normalized to unit power:
\begin{align}
1 = \frac{1}{2} \real\left\{ - \iint\displaylimits_{z=0} \mathcal{E}_m \times \mathcal{H}_m^\ast \cdot \hat{z} \, dx \, dy \right\} .
\end{align}

To determine the amplitude $a$ of the scattered light coupled into the waveguide, we will use a mode overlap integral over the $z = 0$ plane. From \cite{prmcisaac_ieeetmtt1991}, the amplitude of the waveguide mode propagating in the $+z$ direction (i.e. \emph{towards} the output) is given by
\begin{align}
a = \frac{1}{N} \iint\displaylimits_{z=0} \left(\mathbf{E}_m \times \mathbf{H} - \mathbf{E} \times \mathbf{H}_m \right) \cdot \hat{z} \, dx \, dy.
\end{align}
The normalization constant $N$ is
\begin{align}
N &= - 2 \iint\displaylimits_{z=0} \left( \mathcal{E}_m \times \mathcal{H}_m \right) \cdot \hat{z} \, dx \, dy.
\end{align}
We can simply our expression for $N$ by making use of the fact that the transverse modal fields are purely real and normalized to unit power, yielding
\begin{align}
N = - 2 \real \left\{ \iint\displaylimits_{z=0} \left( \mathcal{E}_m \times \mathcal{H}_m^\ast \right) \cdot \hat{z} \, dx \, dy \right\} = 4.
\end{align}
Since the $z = 0$ plane coincides with the surface $S_1$, and $\hat{z} \, dx \, dy = d\mathbf{A}$, we can rewrite our expression for the output amplitude $a$ as
\begin{align}
a = \frac{1}{4} \iint\displaylimits_{S_1} \left(\mathbf{E}_m \times \mathbf{H} - \mathbf{E} \times \mathbf{H}_m \right) \cdot d\mathbf{A}.
\label{eqn:wg_overlap}
\end{align}
Finally, since the modal fields are power normalized, the collected power $P$ is given by
\begin{align}
P = |a|^2.
\label{eqn:pwr_ampl}
\end{align}

\subsection{Applying reciprocity}
We now have an expression (\ref{eqn:wg_overlap}) for the output amplitude $a$ in terms of the scattered fields coupled into the waveguide. Unfortunately, this expression is difficult to use directly, since it requires us to model the propagation of scattered light through the lens and the waveguide coupler. We would much prefer to have a field overlap integral over surface $S_2$ immediately adjacent to the scattering surface, where we can neglect everything except the coherence properties of the scattered light. We will achieve this by making use of the reciprocity theorem.

For a closed surface $S$ that encloses no sources, the Lorentz reciprocity theorem \cite{rfharrington_2001_reciprocity} states that
\begin{align}
0 = \oiint\displaylimits_{S} \left( \mathbf{E}_m \times \mathbf{H} - \mathbf{E} \times \mathbf{H}_m \right) \cdot d\mathbf{S}.
\label{eqn:reciprocity}
\end{align}
Earlier, we assumed that the modal fields $\mathbf{E}_m$ and $\mathbf{H}_m$ are zero everywhere on $S$ except for the front and back surfaces $S_1$ and $S_2$. This allows us to rewrite (\ref{eqn:reciprocity}) as
\begin{align}
\iint\displaylimits_{S_1} \left( \mathbf{E}_m \times \mathbf{H} - \mathbf{E} \times \mathbf{H}_m \right) \cdot d\mathbf{S}
= - \iint\displaylimits_{S_2} \left( \mathbf{E}_m \times \mathbf{H} - \mathbf{E} \times \mathbf{H}_m \right) \cdot d\mathbf{S}.
\label{eqn:reciprocity_1}
\end{align}
The term on the left side is exactly the same as the mode overlap integral (\ref{eqn:wg_overlap}) for the output amplitude $a$, except for a factor of $\frac{1}{4}$. The output amplitude is thus given by the overlap integral
\begin{align}
a = \frac{1}{4} \iint\displaylimits_{S_2} \left( \mathbf{E} \times \mathbf{H}_m - \mathbf{E}_m \times \mathbf{H} \right) \cdot d\mathbf{S}.
\label{eqn:s2_overlap}
\end{align}
over surface $S_2$, which is adjacent to the scattering surface.

At this point in our discussion, it is convenient to define the \emph{efficiency} $\eta$ of the optical system as the fraction of modal power reaching the target surface $S_2$: 
\begin{align}
\eta = \frac{1}{2} \real \left\{ \iint\displaylimits_{S_2} \left( \mathbf{E}_m \times \mathbf{H}_m^\ast \right) \cdot d\mathbf{S} \right\} 
\label{eqn:efficiency}
\end{align}
In a properly designed system, $\eta$ will be close to unity. If $\eta < 1$, we will see that this is manifested as a reduction in the collected power.

\subsection{``Flat surface'' approximation}
We can typically choose the integration surface $S_2$ to be flat on wavelength length scales. We will use this fact to considerably simplify our field overlap integral (\ref{eqn:s2_overlap}). Consider a locally flat patch of $S_2$:

\begin{figure}[H]
\centering
\includegraphics[scale=0.9]{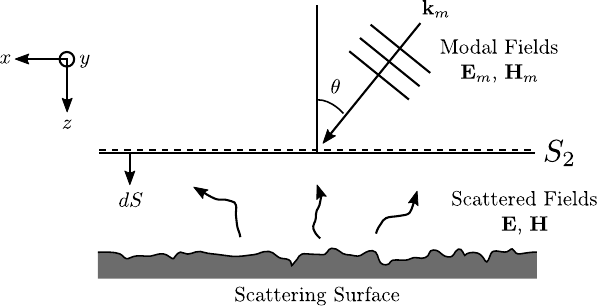}
\caption{We assume that our integration surface $S_2$ is locally flat. In addition, we assume that the modal fields are roughly collimated, with an incidence angle of approximately $\theta$ with respect to the surface $S_2$. \label{fig:scalar_approx}}
\end{figure}

\noindent
We have chosen our local coordinate system such that $z = 0$ on the surface $S_2$, and the surface normal $\hat{z}$ points towards the scattering surface. In addition, we assume that the modal fields $\mathbf{E}_m$, $\mathbf{H}_m$ are roughly collimated in this patch of $S_2$ with an incidence angle of $\theta$ with respect to the surface.

Our goal is to simplify the overlap integral (\ref{eqn:s2_overlap}), which can be written in this locally flat region as
\begin{align}
a = \frac{1}{4} \iint\displaylimits_{z = 0} \left( \mathbf{E} \times \mathbf{H}_m - \mathbf{E}_m \times \mathbf{H} \right) \cdot \hat{z} \, dx \, dy.
\label{eqn:s2_overlap_1}
\end{align}
To do so, we will need to decompose the fields into their constituent plane waves. The scattered fields $\mathbf{E}$ and $\mathbf{H}$ can be written in terms of their plane wave components as
\begin{align}
\mathbf{E}(\mathbf{r}) 
&= \iint\displaylimits_{-\infty}^{\infty} \hat{\mathbf{E}}(k_x, k_y) \; e^{-j \mathbf{k} \cdot \mathbf{r}} \, dk_x \, dk_y 
\label{eqn:E_decomp} \\
\mathbf{H}(\mathbf{r}) 
&= \frac{1}{Z_0} \iint\displaylimits_{-\infty}^{\infty} \hat{k} \times \hat{\mathbf{E}}(k_x, k_y) \; e^{-j \mathbf{k} \cdot \mathbf{r}} \, dk_x \, dk_y.
\label{eqn:H_decomp}
\end{align}
Here,  $\mathbf{r} = (x, y, z)$ is the position, the wavevector $\mathbf{k}$ is given by
\begin{align}
\mathbf{k} = \left(k_x, \, k_y, \, \sqrt{\frac{4 \pi^2}{\lambda^2} - k_x^2 - k_y^2} \right),
\label{eqn:k_scattered}
\end{align}
${\hat{k} = \mathbf{k} / |\mathbf{k}|}$ is the propagation direction, $Z_0$ is the impedance of free space, and $\lambda$ is the wavelength of light. We can see from equation (\ref{eqn:k_scattered}) that the scattered waves only propagate in the $-z$ direction, as desired.

Similarly, the modal fields $\mathbf{E}_m$ and $\mathbf{H}_m$ can be decomposed as 
\begin{align}
\mathbf{E}_m(\mathbf{r}) 
&= \iint\displaylimits_{-\infty}^{\infty} \hat{\mathbf{E}}_m(k_x, k_y) \; e^{-j \mathbf{k}_m \cdot \mathbf{r}} \, dk_x \, dk_y 
\label{eqn:Em_decomp} \\
\mathbf{H}_m(\mathbf{r}) 
&= \frac{1}{Z_0} \iint\displaylimits_{-\infty}^{\infty} \hat{k}_m \times \hat{\mathbf{E}}_m(k_x, k_y) \; e^{-j \mathbf{k}_m \cdot \mathbf{r}} \, dk_x \, dk_y,
\label{eqn:Hm_decomp}
\end{align}
where the wavevector $\mathbf{k}_m$ is given by
\begin{align}
\mathbf{k}_m = \left(k_x, \, k_y, \, \sqrt{\frac{4 \pi^2}{\lambda^2} - k_x^2 - k_y^2} \right).
\label{eqn:k_modal}
\end{align}
Note that the modal fields propagate in the $+z$ direction, in contrast to the scattered fields.

Now, recall that the overlap intgegral of two functions $f(x)$ and $g(x)$ can be written in terms of their Fourier transforms $\hat{f}(k)$ and $\hat{g}(k)$ as
\begin{align}
\int\displaylimits_{-\infty}^{\infty} f(x) \, g(x) \, dx
= \int\displaylimits_{-\infty}^{\infty} \hat{f}(-k) \, \hat{g}(k) \, dk
\label{eqn:plancherel}
\end{align}
due to Plancherel's theorem. Applying (\ref{eqn:plancherel}) and the plane wave decompositions (\ref{eqn:E_decomp}) and (\ref{eqn:Hm_decomp}) to the first term in the overlap integral (\ref{eqn:s2_overlap_1}) yields
\begin{align}
& \frac{1}{4} \iint\displaylimits_{z=0} \left( \mathbf{E} \times \mathbf{H}_m \right) \cdot \hat{z} \, dx \, dy \nonumber \\
&= \frac{1}{4 \, Z_0} \iint\displaylimits_{-\infty}^{\infty}
\left[ \hat{\mathbf{E}}(-k_x, -k_y) \times \left( \hat{k}_m(k_x, k_y) \times \hat{\mathbf{E}}_m(k_x, k_y) \right) \right]
\cdot \hat{z} \, dk_x \, dk_y.
\label{eqn:s2_overlap_term1_1}
\end{align}
If we make use of the vector triple product $\mathbf{a} \times (\mathbf{b} \times \mathbf{c}) = \mathbf{b} (\mathbf{a} \cdot \mathbf{c)} - \mathbf{c} (\mathbf{a} \cdot \mathbf{b})$, our expression becomes
\begin{align}
& \frac{1}{4} \iint\displaylimits_{z=0} \left( \mathbf{E} \times \mathbf{H}_m \right) \cdot \hat{z} \, dx \, dy \nonumber \\
&= \frac{1}{4 \, Z_0} \iint\displaylimits_{-\infty}^{\infty} \left[
\left(\hat{\mathbf{E}}(-k_x, -k_y) \cdot \hat{\mathbf{E}}_m(k_x, k_y) \right) 
\left( \hat{k}_m(k_x, k_y) \cdot \hat{z} \right) \right. \nonumber \\
&\qquad - \left. \left( \hat{\mathbf{E}}(-k_x, -k_y) \cdot \hat{k}_m(k_x, k_y) \right)
\left( \hat{\mathbf{E}}_m(k_x, k_y) \cdot \hat{z} \right) \right]
 \, dk_x \, dk_y. 
\label{eqn:s2_overlap_term1_2}
\end{align}
However, we can see from (\ref{eqn:k_scattered}) and (\ref{eqn:k_modal}) that $\hat{k}_m(k_x, k_y) = - \hat{k}(-k_x, -k_y)$, and the plane wave amplitude $\hat{E}(k_x, k_y)$ is always perpendicular to its propagation direction $\hat{k}(k_x, k_y)$, so the second term in (\ref{eqn:s2_overlap_term1_2}) disappears. We are left with
\begin{align}
& \frac{1}{4} \iint\displaylimits_{z=0} \left( \mathbf{E} \times \mathbf{H}_m \right) \cdot \hat{z} \, dx \, dy \nonumber \\
&= \frac{1}{4 \, Z_0} \iint\displaylimits_{-\infty}^{\infty}
\left(\hat{\mathbf{E}}(-k_x, -k_y) \cdot \hat{\mathbf{E}}_m(k_x, k_y) \right) 
\left( \hat{k}_m(k_x, k_y) \cdot \hat{z} \right)
 \, dk_x \, dk_y. 
\label{eqn:s2_overlap_term1_3}
\end{align}
We now make use of our assumption that the modal beam is roughly collimated with an incidence angle of $\theta$, which allows us to approximate
\begin{align}
\hat{k}_m(k_x, k_y) \cdot \hat{z} \approx \cos\theta.
\end{align}
Our final expression in the spatial frequency domain is then
\begin{align}
& \frac{1}{4} \iint\displaylimits_{z=0} \left( \mathbf{E} \times \mathbf{H}_m \right) \cdot \hat{z} \, dx \, dy \nonumber \\
&= \frac{1}{4 \, Z_0} \iint\displaylimits_{-\infty}^{\infty}
\left(\hat{\mathbf{E}}(-k_x, -k_y) \cdot \hat{\mathbf{E}}_m(k_x, k_y) \right) 
\cos \theta
 \, dk_x \, dk_y. 
\label{eqn:s2_overlap_term1_4}
\end{align}
Transforming back to real space, we obtain
\begin{align}
\frac{1}{4} \iint\displaylimits_{z=0} \left( \mathbf{E} \times \mathbf{H}_m \right) \cdot \hat{z} \, dx \, dy
= \frac{1}{4 \, Z_0} \iint\displaylimits_{-\infty}^{\infty}
\mathbf{E} \cdot \mathbf{E}_m \, \cos \theta \, dx \, dy
\label{eqn:s2_overlap_term1_5}
\end{align}
Similarly, the second term in the overlap integral (\ref{eqn:s2_overlap_1}) becomes
\begin{align}
-\frac{1}{4} \iint\displaylimits_{z=0} \left( \mathbf{E}_m \times \mathbf{H} \right) \cdot \hat{z} \, dx \, dy
= \frac{1}{4 \, Z_0} \iint\displaylimits_{-\infty}^{\infty}
\mathbf{E} \cdot \mathbf{E}_m \, \cos \theta \, dx \, dy
\label{eqn:s2_overlap_term2_1}
\end{align}
The overlap integral (\ref{eqn:s2_overlap_1}) for the output amplitude $a$ is thus approximately
\begin{align}
a = \frac{1}{2\, Z_0} \iint\displaylimits_{S_2} \mathbf{E} \cdot \mathbf{E}_m \, \cos \theta \, dS
\label{eqn:s2_overlap_simple}
\end{align}
This expression will be accurate as long as the integration surface $S_2$ is locally smooth on the wavelength scale, and the modal beam is roughly collimated over regions the size of several wavelengths.

\subsection{Uniform illumination}
The expected value of the recieved power depends upon the statistical properties of the scattered electric field $\mathbf{E}$. To make the problem tractable, we first consider the simple case of a flat and uniform scattering surface that is illuminated uniformly across the entire surface.

More precisely, our assumptions are the following:
\begin{enumerate}
\item $S_2$ is a completely flat surface. As in the previous section, we choose our coordinate system such that $S_2$ coincides with the $z = 0$ plane.

\item The modal field is roughly collimated, i.e. $\cos \theta$ is approximately constant across all of $S_2$.

\item The modal field $\mathbf{E}_m$ has a fixed polarization $\hat{p}$ and has a slowly varying amplitude across $S_2$. The modal electric field on the $z=0$ plane is then of the form
\begin{align}
\mathbf{E}_m(x, y, 0) = \psi(x, y) \, e^{-j (k_x x + k_y y)} \hat{p}
\label{eqn:sva_E}
\end{align}
where the slowly varying scalar amplitude $\psi(x, y)$ is constant over wavelength length scales, and $|\hat{p}|^2 = 1$.

\item The scattered field $\mathbf{E}$ has uniform intensity and coherence width over $S_2$. More precisely, the autocorrelation function $\mathcal{W}(\mathbf{r}_1, \mathbf{r}_2)$ of the electric field $\mathbf{E}(\mathbf{r})$ over $S_2$ only depends on the difference $\mathbf{r}_1 - \mathbf{r}_2$ in the positions. Explicitly,
\begin{align}
\braket{\mathbf{E}(x_1, y_1, 0) \, \mathbf{E}^\dagger(x_2, y_2, 0)} = \mathcal{W}(x_1 - x_2, y_1 - y_2),
\label{eqn:def_autocorr}
\end{align}
where $\braket{\cdot}$ indicates the expected value, and $\dagger$ is the complex transpose. In general, the electric field autocorrelation $\mathcal{W}$ will be a $3 \times 3$ tensor.

\end{enumerate}
Some of these assumptions will be relaxed later. 

Using the assumptions $1 - 3$, our expression (\ref{eqn:s2_overlap_simple}) for the output amplitude $a$ becomes
\begin{align}
a = 
\frac{\cos \theta}{2\, Z_0} \iint\displaylimits_{-\infty}^{\infty}
\left( \mathbf{E}(x, y, 0) \cdot \hat{p} \right)
\psi(x, y) \,  e^{-j (k_x x + k_y y)}  \, dx \, dy.
\label{eqn:s2_overlap_simple_unif}
\end{align}
The expected value of the collected power $P$ is then
\begin{align}
\Braket{P}
&= \Braket{|a|^2} \nonumber \\
&= 
 \frac{\cos^2 \theta}{4\, Z_0^2} \iiiint\displaylimits_{-\infty}^{\infty}
 \bar{p}^\dagger \Braket{\mathbf{E}(x_1 , y_1, 0) \mathbf{E}(x_2 , y_2, 0)^\dagger} \bar{p}
 \  e^{-j (k_x (x_1 - x_2) + k_y (y_1 - y_2))} \nonumber \\
& \qquad\qquad\qquad \times  
 \psi(x_1, y_1) \, \psi^\ast(x_2, y_2)
 \, dx_1 \, dx_2 \, dy_1 \, dy_2 
\label{eqn:P_general_0}
\end{align}
where $\bar{p} = \hat{p}^\ast$. Making use of our definition (\ref{eqn:def_autocorr}) for the electric field autocorrelation $\mathcal{W}$ yields
\begin{align}
\Braket{P}
&= 
 \frac{\cos^2 \theta}{4\, Z_0^2} \iiiint\displaylimits_{-\infty}^{\infty}
 \bar{p}^\dagger \, \mathcal{W}(x_1 - x_2, y_1 - y_2) \, \bar{p}
 \  e^{-j (k_x (x_1 - x_2) + k_y (y_1 - y_2))} \nonumber \\
& \qquad\qquad\qquad \times  
 \psi(x_1, y_1) \, \psi^\ast(x_2, y_2)
 \, dx_1 \, dx_2 \, dy_1 \, dy_2 
\label{eqn:P_general_1}
\end{align}

To further simplify our expression, we define the \emph{scalar autocorrelation function} $\rho(x,y)$ as
\begin{align}
\rho(x, y) = 
 \frac{\cos \theta}{2\, Z_0}
 \; \bar{p}^\dagger \, \mathcal{W}(x, y) \, \bar{p}
 \  e^{-j (k_x x + k_y y)},
\end{align}
and our expression for the collected power becomes
\begin{align}
\Braket{P}
&= 
  \frac{\cos \theta}{2\, Z_0}
  \iiiint\displaylimits_{-\infty}^{\infty}
  \rho(x_1 - x_2, y_1 - y_1)
  \, \psi(x_1, y_1) \, \psi^\ast(x_2, y_2)
 \, dx_1 \, dx_2 \, dy_1 \, dy_2 
\label{eqn:P_general_2}
\end{align}
The scalar autocorrelation function $\rho(x, y)$ vanishes for large $x$ and $y$; this length scale is known as the \emph{coherence width}. Since $\psi(x, y)$ is a slowly varying amplitude, $\rho$ will typically vanish on much shorter length scales than on which $\psi$ varies. This allows us to greatly simplify our expression for $P$. More precisely, if $\rho(x, y) \approx 0$ for any $\sqrt{x^2 + y^2} > L$, and $\psi(x,y)$ is effectively constant on length scale $L$, we can approximate $\rho$ as a delta function:
\begin{align}
\rho(x, y) \approx P_0 \, \delta(x) \, \delta(y)
\label{eqn:rho_approx_unif}
\end{align}
where the characteristic power $P_0$ of a single speckle is
\begin{align}
P_0 
= \iint\displaylimits_{-\infty}^{\infty} \rho(x, y) \, dx \, dy
= \frac{\cos \theta}{2\, Z_0}
 \; \bar{p}^\dagger
 \left[ \;
 \iint\displaylimits_{-\infty}^{\infty}
 \, \mathcal{W}(x, y) \,
 \  e^{-j (k_x x + k_y y)}
 \, dx \, dy
 \right]
 \bar{p}
\label{eqn:def_speckle_power}
\end{align}
However, the term in the brackets is equal to the spatial Fourier transform of the autocorrelation tensor, 
\begin{align}
\hat{\mathcal{W}}(k_x, k_y) 
= \iint\displaylimits_{-\infty}^{\infty}
 \mathcal{W}(x,y) \, e^{-j k_x x - j k_y y} \, dx \, dy,
\label{eqn:autocorr_ft}
\end{align}
which allows us to write the characteristic power $P_0$ more concisely as
\begin{align}
P_0 
= \frac{\cos \theta}{2\, Z_0}
 \; \bar{p}^\dagger
 \, \hat{\mathcal{W}}(k_x, k_y) 
 \, \bar{p}.
\label{eqn:speckle_power_simple}
\end{align}
Under this approximation, (\ref{eqn:P_general_2}) becomes
\begin{align}
\Braket{P}
&= 
  \frac{\cos \theta}{2\, Z_0}
  \iiiint\displaylimits_{-\infty}^{\infty}
  P_0 \, \delta(x_1 - x_2) \, \delta(y_1 - y_2) 
  \, \psi(x_1, y_1) \, \psi^\ast(x_2, y_2)
  \, dx_1 \, dx_2 \, dy_1 \, dy_2 
 \nonumber \\
&= 
  P_0 \  \frac{\cos \theta}{2\, Z_0}
  \iint\displaylimits_{-\infty}^{\infty}
 \, \left| \psi(x, y) \right|^2
 \, dx \, dy.
\label{eqn:P_general_3}
\end{align}

Now, recall that our modal field is power normalized, and the power reaching the target surface $S_2$ is given by the efficiency $\eta$: 
\begin{align}
\eta = \frac{1}{2} \real \left\{ \iint\displaylimits_{S_2} \left( \mathbf{E}_m \times \mathbf{H}_m^\ast \right) \cdot d\mathbf{S} \right\} .
\tag{\ref{eqn:efficiency}}
\end{align}
For our ansatz modal E-field (\ref{eqn:sva_E}), the corresponding H-field is given by
\begin{align}
\mathbf{H}_m(x, y, 0) = \frac{1}{Z_0} \, \psi(x, y) \, e^{-j (k_x x + k_y y)} \hat{k}_m \times \hat{p},
\label{eqn:sva_H}
\end{align}
where $\hat{k}_m = \mathbf{k}_m / |\mathbf{k}_m|$ is the propagation direction, and the wavevector $\mathbf{k}_m$ is given by (\ref{eqn:k_modal}). Substituting our ansatz (\ref{eqn:sva_E}) and (\ref{eqn:sva_H}) into our expression for the efficiency (\ref{eqn:efficiency}) yields
\begin{align}
\eta
&= \frac{1}{2 \, Z_0} \real\left\{ \left( \hat{p} \times \left( \hat{k}_m \times \hat{p}^\ast \right) \right) \cdot \hat{z} \right\}
 \iint\displaylimits_{-\infty}^{\infty} \left| \psi(x, y) \right|^2 \, dx \, dy
 \nonumber \\
&= \frac{1}{2 \, Z_0} \real\left\{ \left| \hat{p} \right|^2 \hat{k}_m \cdot \hat{z}
      - ( \hat{p} \cdot \hat{k}_m ) ( \hat{p} \cdot \hat{z} ) \right\}
 \iint\displaylimits_{-\infty}^{\infty} \left| \psi(x, y) \right|^2 \, dx \, dy
 \nonumber \\
&= \frac{\cos \theta}{2 \, Z_0}
 \iint\displaylimits_{-\infty}^{\infty} \left| \psi(x, y) \right|^2 \, dx \, dy
\label{eqn:efficiency_sva}
\end{align}
where we have made use of our normalization $|\hat{p}|^2 = 1$, the orthogonality of the polarization and propagation direction $\hat{p} \cdot \hat{k}_m = 0$, and our definition of the incidence angle $\hat{k}_m \cdot \hat{z} = \cos \theta$.

We obtain our final result by substituting (\ref{eqn:efficiency_sva}) into (\ref{eqn:P_general_3}), which allows us to express the collected power in terms of the efficiency $\eta$ of the system:
\begin{align}
\braket{P} = \eta \, P_0.
\end{align}
Thus, if we have a lossless system with $\eta = 1$, the expected collected power is exactly equal to the characteristic power $P_0$ of a single speckle in the scattered field.

\subsection{Non-uniform illumination}
In the previous section, we assumed that the scattered light is uniform in intensity across the entire scattering surface. This assumption is not true for most coherent LiDAR systems, which typically scan the scene with tightly focused laser beams.

We can model the case of non-uniform illumination by allowing the speckle power $P_0$ to vary across the scattering surface. We will find it particularly illuminating to write the speckle power in terms of the scattered light intensity $I(x, y)$ and the characteristic speckle area $A_0$,
\begin{align}
P_0(x, y) = A_0 \, I(x, y).
\label{eqn:def_speckle_area}
\end{align}
The characteristic area $A_0$ of a speckle is a property of the scattering surface, which we will assume remains constant over the entire scattering surface. Meanwhile, the scattered light intensity $I(x, y)$ is proportional to the intensity of the illuminating light, and can vary strongly with position. 

Under these assumptions, our expression \ref{eqn:P_general_3} for the expected power $\braket{P}$ becomes 
\begin{align}
\braket{P}
&= A_0 \frac{\cos \theta}{2\, Z_0}
  \iint\displaylimits_{-\infty}^{\infty}
 I(x, y)
 \, \left| \psi(x, y) \right|^2
 \, dx \, dy.
\label{eqn:P_nonunif_0}
\end{align}
The received power is therefore entirely dependent upon the overlap integral between the scattered light intensity $I(x, y)$ and modal field $\left| \psi(x, y) \right|^2$.

For convenience, we define the power normalized modal field $\bar{\psi}(x, y)$ as
\begin{align}
\bar{\psi}(x, y) = \sqrt{\frac{\cos \theta}{2 \, Z_0 \, \eta}}
 \; \psi(x, y),
\end{align}
which satisfies the property
\begin{align}
1 = 
 \iint\displaylimits_{-\infty}^{\infty}
 \left| \bar{\psi}(x, y) \right|^2 
 \, dx \, dy.
\end{align}
Meanwhile, the total power of the scattered light $P_s$ is given by
\begin{align}
 P_s = \iint\displaylimits_{-\infty}^{\infty}
 I(x, y)
 \, dx \, dy,
\end{align}
which can be used to define a normalized scattered light intensity $\bar{I}(x, y)$,
\begin{align}
\bar{I}(x, y) = \frac{I(x, y)}{P_s}.
\label{eqn:def_intensity_norm}
\end{align}
Rewriting \ref{eqn:P_nonunif_0} in terms of the normalized quantities $\bar{\psi}(x, y)$ and $\bar{I}(x, y)$ yields our final expression for the expected power,
\begin{align}
\braket{P}
&= \eta \, P_s \, A_0 
  \iint\displaylimits_{-\infty}^{\infty}
 I(x, y)
 \, \left| \bar{\psi}(x, y) \right|^2
 \, dx \, dy.
 \label{eqn:P_nonunif_1}
\end{align}

\subsubsection{Monostatic LiDAR}
Many coherent LiDAR designs are \emph{monostatic}, sharing the transmit and receive paths. For example, a monostatic LiDAR may use the same collimator and optical fiber to both transmit and receive light. In the case of a monostatic LiDAR, the intensity distribution of the scattered light $\bar{I}(x, y)$ and the collection mode $|\bar{\psi}(x, y)|^2$ will be identical, i.e. 
\begin{align}
\bar{I}(x, y) = \left| \bar{\psi}(x, y) \right|^2.
\end{align}
If we define the effective area $A_\mathit{eff}$ of the laser beam as 
\begin{align}
\frac{1}{A_\mathit{eff}} 
 = \iint\displaylimits_{-\infty}^{\infty}
 \bar{I}^2(x, y),
 \, dx \, dy,
 \label{eqn:def_eff_area}
\end{align}
we obtain a very simple and intuitive expression for the expected power,
\begin{align}
\braket{P}  = \eta \, P_s \, \frac{A_0}{A_\mathit{eff}}.
 \label{eqn:P_monostatic}
\end{align}
The collected power $\braket{P}$ can thus be understood as the fraction of scattered power $P_s$ that falls within a characteristic speckle area $A_0$, multipled by the efficiency $\eta$ of the collection optics.

We will now compute the effective beam area for several commonly encountered beam profiles. 

\paragraph*{Flat-top beam} \mbox{} \newline
A flat-top beam has an intensity profile of the form
\begin{align}
\bar{I}(x,y) 
    = \frac{1}{A_\mathit{beam}}
    \begin{cases}
        1, & (x, y) \in \Omega \\
        0, & \mathrm{otherwise,}
    \end{cases}
\end{align}
where $A_\mathit{beam}$ is the physical area of the beam. From equation \ref{eqn:def_eff_area}, we can see that the effective beam area $A_\mathit{eff}$ is exactly equal to the physical area of the beam, since
\begin{align}
\frac{1}{A_\mathit{eff}} 
 &= \iint\displaylimits_{-\infty}^{\infty}
     \bar{I}^2(x, y)
     \, dx \, dy
     \nonumber \\
 &= \iint\displaylimits_{\Omega}
     \frac{1}{A_\mathit{beam}^2}
     \, dx \, dy
     \nonumber \\
 &= \frac{1}{A_\mathit{beam}}.
\end{align}

\paragraph*{Gaussian beam} \mbox{} \newline
A Gaussian beam with beam radius $w$ has an intensity profile of the form
\begin{align}
\bar{I}(x,y) 
    = N \exp\left( - \frac{2 (x^2 + y^2)}{w^2} \right).
\end{align}
To find the normalization $N$, we take advantage of the fact that the integral of $\bar{I}(x,y)$ should be unity:
\begin{align}
1 &= \iint\displaylimits_{-\infty}^{\infty}
     \bar{I}(x, y)
     \, dx \, dy
     \nonumber \\
  &= \iint\displaylimits_{-\infty}^{\infty}
     N \exp\left( - \frac{2 (x^2 + y^2)}{w^2} \right)
     \, dx \, dy
     \nonumber \\
  &= \frac{N \, \pi \, w^2}{2}.
\end{align}
Thus, the normalization constant $N = 2 / \pi w^2$. Finding the effective area using equation \ref{eqn:def_eff_area} then yields
\begin{align}
\frac{1}{A_\mathit{eff}}
   &= \iint\displaylimits_{-\infty}^{\infty}
     \bar{I}^2(x, y)
     \, dx \, dy
     \nonumber \\
  &= \frac{4}{\pi^2 w^4} \iint\displaylimits_{-\infty}^{\infty}
     \exp\left( - \frac{4 (x^2 + y^2)}{w^2} \right).
     \, dx \, dy
     \nonumber \\
  &= \frac{1}{\pi w^2}.
\end{align}
We can conclude that the effective area of a Gaussian beam is simply
\begin{align}
A_\mathit{eff} = \pi w^2.
\end{align}

\subsection{Ideal Lambertian scattering surface}
We will now consider the case of Lambertian scattering, which is an idealized model of diffuse scattering. An Lambertian scatterer takes all incident light and scatters it in a cosine distribution, as illustrated in figure \ref{fig:lambertian_scattering}.
\begin{figure}[h!tb]
\centering
\includegraphics[scale=0.9]{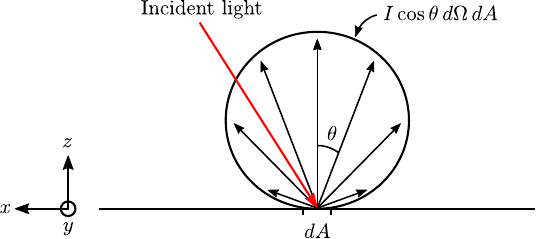}
\caption{In Lambertian scattering, the incident light is scattered into a cosine distribution. The percieved brightness of a Lambertian scatterer is \emph{independent} of the viewing angle, and only depends upon the incident flux per unit area of the surface. \label{fig:lambertian_scattering}}
\end{figure}
The angular distribution of the scattered light is independent of the incidence angle of the incident light.

Rather conveniently for us, an ideal blackbody emitter is a Lambertian emitter. In other words, the light scattered by a Lambertian scatterer is indistinguishable from blackbody radiation, provided that the scattered light is completely depolarized. We can thus directly apply previous results for the spatial coherence of blackbody radiation to Lambertian scattering.

From \cite{kblomstedt_pra2013}, the electric field at the surface of a perfect blackbody is given by the superposition of planewaves
\begin{align}
\mathbf{E}(\mathbf{r}) = \int\displaylimits_\alpha \mathbf{A}(\hat{u}) \, \exp(-j k_0 \hat{u} \cdot \mathbf{r}) \, d\Omega
\label{eqn:blackbody_E}
\end{align}
where
\begin{align}
\hat{u} = \sin\theta \cos\phi \, \hat{x} + \sin\theta \sin\phi \, \hat{y} + \cos\theta \, \hat{z}
\label{eqn:def_u_hat}
\end{align}
is the unit vector that specifies the propagation direction of a plane wave in spherical coordinates $(\phi, \theta)$. The differential solid angle is given by $d\Omega = \sin\theta \, d\theta \, d\phi$, and $\alpha$ is the $2 \pi$ solid angle corresponding to the upper hemisphere. The wavenumber $k_0$ in the medium is related to the wavelength $\lambda$ by $k_0 = 2 \pi / \lambda$.

The vector amplitude $\mathbf{A}(\hat{u})$ of an individual planewave can be decomposed in terms of the s and p polarizations as
\begin{align}
\mathbf{A}(\hat{u}) = A_s(\hat{u}) \hat{s} + A_p(\hat{u}) \hat{p}
\label{eqn:s_p_decomp}
\end{align}
where the polarization vectors $\hat{s}$ and $\hat{p}$ are given by
\begin{align}
\hat{s} &= \frac{\hat{z} \times \hat{u}}{|\hat{z} \times \hat{u}|}
 \label{eqn:def_s_hat} \\
\hat{p} &= \hat{s} \times \hat{u}.
 \label{eqn:def_p_hat}
\end{align}
The correlations between the amplitudes are
\begin{align}
\braket{A_s^\ast (\hat{u}_1) \, A_s (\hat{u}_2 )}
 &= \braket{A_p^\ast (\hat{u}_1) \, A_p (\hat{u}_2 )}
 = a_0 \delta(\hat{u}_1 - \hat{u}_2))
 \label{eqn:selfcorr_s_p}
 \\
\braket{A_s^\ast (\hat{u}_1) \, A_p (\hat{u}_2 )}
 &= \braket{A_p^\ast (\hat{u}_1) \, A_s (\hat{u}_2 )}
 = 0,
 \label{eqn:crosscorr_s_p}
\end{align}
where $a_0$ is real and positive, and the directional Dirac delta function is defined as
\begin{align}
\delta(\hat{u}_1 - \hat{u}_2)
= \frac{\delta(\phi_2 - \phi_1) \, \delta(\theta_2 - \theta_1) }{\sin \theta_2}.
\label{eqn:def_direc_diracdelta}
\end{align}

\subsubsection{Fourier transform of electric field autocorrelation}
Fundamentally, we are interested in finding the characteristic power $P_0$ of a speckle. We can see from equation (\ref{eqn:speckle_power_simple}) that this requires us to take the spatial Fourier transform of the electric field autocorrelation tensor $\mathcal{W}$, defined as
\begin{align}
\hat{\mathcal{W}}(k_x, k_y) 
= \iint\displaylimits_{-\infty}^{\infty}
 \mathcal{W}(x,y) \, e^{-j k_x x - j k_y y} \, dx \, dy
 \tag{\ref{eqn:autocorr_ft}}
\end{align}
Meanwhile, from \cite{kblomstedt_pra2013}, the electric field autocorrelation tensor at the surface of a perfect blackbody is 
\begin{align}
\mathcal{W}(x,y) 
= a_0 \int\displaylimits_\alpha 
  \left( \mathcal{U}_3 - \hat{u} \hat{u}^T \right) 
  e^{jk_x x + jk_y y} \, d\Omega.
 \label{eqn:blackbody_corr}
\end{align}
where $\mathcal{U}_3$ is the $3 \times 3$ identity matrix. Our goal, then, is to change the integration variables in (\ref{eqn:blackbody_corr}) to $k_x$ and $k_y$ so that it can be understood as an inverse spatial Fourier transform. We note that 
\begin{align}
k_x &= k_0 \, \sin\theta \, \cos\phi \\
k_y &= k_0 \, \sin\theta \, \sin\phi.
\end{align}
The Jacobian $J$ of the transformation $(\theta, \phi) \rightarrow (k_x, k_y)$ is then given by
\begin{align}
J^{-1} = 
\begin{pmatrix}
 \dfrac{\partial k_x}{\partial \theta} &  \dfrac{\partial k_x}{\partial \phi} \\[1em]
 \dfrac{\partial k_y}{\partial \theta} &  \dfrac{\partial k_y}{\partial \phi}
\end{pmatrix}
=
k_0 
\begin{pmatrix}
\cos\theta \cos\phi & -\sin\theta \sin\phi \\
\cos\theta \sin\phi & \sin\theta \cos\phi
\end{pmatrix},
\end{align}
and the absolute value of the determinant is
\begin{align}
\left| J \right|^{-1}
= k_0^2  \left(\cos\theta \sin\theta \cos^2\phi + \cos\theta \sin\theta \sin^2\phi \right)
= k_0^2 \cos\theta \sin\theta.
\end{align}
Applying this change of integration variables to (\ref{eqn:blackbody_corr}) yields
\begin{align}
\mathcal{W}(x,y) 
&= a_0 \iint\displaylimits_{k_x^2 + k_y^2 < k_0^2}
  \sin\theta \left| J \right|
  \left( \mathcal{U}_3 - \hat{u} \hat{u}^T \right) 
  e^{jk_x x + jk_y y} \,
  dk_x \, dk_y 
 \nonumber\\
&= \frac{a_0}{k_0^2} \iint\displaylimits_{k_x^2 + k_y^2 < k_0^2}
  \frac{1}{\cos\theta}
  \left( \mathcal{U}_3 - \hat{u} \hat{u}^T \right) 
  e^{jk_x x + jk_y y} \,
  dk_x \, dk_y .
 \label{eqn:blackbody_corr_1}
\end{align}
Applying the Fourier inversion theorem, we conclude that the spatial Fourier transform $\hat{\mathcal{W}}(k_x, k_y)$ of the autocorrelation tensor is
\begin{align}
\hat{\mathcal{W}}(k_x, k_y)
&= 
\begin{cases}
4 \pi^2 \, \dfrac{a_0}{k_0^2} \dfrac{1}{\cos\theta} \left( \mathcal{U}_3 - \hat{u}\hat{u}^T \right), & k_x^2 + k_y^2 \leq k_0^2 \\
0, & \text{otherwise}. 
\end{cases}
\label{eqn:blackbody_autocorr_ft}
\end{align}

\subsubsection{Intensity of blackbody radiation}
Our next step is to relate the normalization constant $a_0$ to the intensity of the blackbody radiation. We define the intensity $I$ as the expected Poynting flux per unit area of the emitting surface:
\begin{align}
I = \left\langle
        \frac{1}{2} \real\left\{
            \mathbf{E}(\mathbf{r}) \times \mathbf{H}^\ast(\mathbf{r})
        \right\}
        \cdot \mathbf{z}
    \right\rangle.
 \label{eqn:def_intensity}
\end{align}
To compute the intensity, we first need to find the magnetic field $\mathbf{H}$ of the blackbody radiation, which is given by
\begin{align}
\mathbf{H}(\mathbf{r})
&= \frac{1}{Z_0} \int\displaylimits_\alpha 
 \left( 
 A_s(\hat{u}) \, \hat{u}\times\hat{s}
 + A_p(\hat{u}) \, \hat{u}\times\hat{p}
 \right)
 \exp(-j k_0 \hat{u} \cdot \mathbf{r}) \, d\Omega 
 \nonumber \\
&= \frac{1}{Z_0} \int\displaylimits_\alpha 
 \left(
 - A_s(\hat{u}) \, \hat{p}
 + A_p(\hat{u}) \, \hat{s}
 \right)
 \exp(-j k_0 \hat{u} \cdot \mathbf{r}) \, d\Omega .
\label{eqn:blackbody_H}
\end{align}
Combining (\ref{eqn:blackbody_E}) and (\ref{eqn:def_intensity}) - (\ref{eqn:blackbody_H}) yields
\begin{align}
I &=  \frac{1}{2 Z_0}
 \real\left\{
  \int\displaylimits_\alpha d\Omega_1   \int\displaylimits_\alpha d\Omega_2 
   \, \exp\left(-j k_0 (\hat{u}_1 - \hat{u}_2) \right)
 \right. \nonumber \\
& \qquad\qquad\qquad
 \left.
  \vphantom{\iint\displaylimits_\alpha}
  \hat{z} \cdot
  \Braket{
    \left(
     A_s(\hat{u}_1) \hat{s}_1 + 
     A_p(\hat{u}_1) \hat{p}_1
    \right)
    \times
    \left(
     - A_s^\ast(\hat{u_2}) \, \hat{p}_2
     + A_p^\ast(\hat{u_2}) \, \hat{s}_2
    \right)
  }
 \right\}.
\end{align}
Simplifying using the amplitude correlations (\ref{eqn:selfcorr_s_p}) and (\ref{eqn:crosscorr_s_p}), we obtain
\begin{align}
I &=  \frac{1}{2 Z_0}
 \real\left\{
  \int\displaylimits_\alpha d\Omega_1   \int\displaylimits_\alpha d\Omega_2 
   \, \exp\left(-j k_0 (\hat{u}_1 - \hat{u}_2) \right)
 \right. \nonumber \\
& \qquad\qquad\qquad
 \left.
  \vphantom{\int\displaylimits_\alpha}
  \hat{z} \cdot
  \Braket{
    A_s(\hat{u}_1) A_s^\ast(\hat{u}_2) \hat{s}_1 \times \hat{p}_2
    +
    A_p(\hat{u}_1) A_p^\ast(\hat{u}_2) \hat{s}_2 \times \hat{p}_1
  }
 \right\}
\nonumber \\
&=  \frac{1}{2 Z_0}
 \real\left\{
  \int\displaylimits_\alpha d\Omega_1   \int\displaylimits_\alpha d\Omega_2 
   \, \exp\left(-j k_0 (\hat{u}_1 - \hat{u}_2) \right)
 \right. \nonumber \\
& \qquad\qquad\qquad
 \left.
  \vphantom{\int\displaylimits_\alpha}
  a_0 \delta(\hat{u}_1 - \hat{u}_2)
  \, \hat{z} \cdot \left( \hat{s}_1 \times \hat{p}_2 +  \hat{s}_2 \times \hat{p}_1 \right)
 \right\}
\nonumber \\
&=  \frac{a_0}{Z_0}
  \int\displaylimits_\alpha d\Omega_1
  \, \real\left\{ \hat{z} \cdot \left( \hat{s}_1 \times \hat{p}_1 \right) \right\}.
\end{align}
However, since $(\hat{p} \times \hat{s}) \cdot \hat{z} = \hat{u} \cdot \hat{z} = \cos\theta$, 
\begin{align}
I =  \frac{a_0}{Z_0}
  \int\displaylimits_\alpha 
  \cos\theta_1 \, d\Omega_1.
\end{align}
This integral over the $2\pi$ solid angle is
\begin{align}
\int\displaylimits_\alpha \cos\theta \, d\Omega 
= \int\displaylimits_{0}^{2\pi} d\phi \int\displaylimits_{0}^{\pi/2} d\theta \, \sin\theta \, \cos\theta
= \pi,
\end{align}
yielding our final expression for the intensity
\begin{align}
I = \frac{\pi a_0 }{Z_0}.
\label{eqn:blackbody_I}
\end{align}

\subsubsection{Speckle power and area}
We are now in a position to compute the characteristic power $P_0$ of a single speckle, which we defined as
\begin{align}
P_0 
= \frac{\cos \theta}{2\, Z_0}
 \; \bar{p}^\dagger
 \, \hat{\mathcal{W}}(k_x, k_y) 
 \, \bar{p}.
\tag{\ref{eqn:speckle_power_simple}}
\end{align}
Substituting our result for the spatial Fourier transform of the autocorrelation tensor $\hat{\mathcal{W}}(k_x, k_y)$ for blackbody radiation (\ref{eqn:blackbody_autocorr_ft}) yields
\begin{align}
P_0 
&= \frac{\cos \theta}{2\, Z_0} 4 \pi^2 \, \frac{a_0}{k_0^2} \frac{1}{\cos\theta}
 \; \bar{p}^\dagger \left( \mathcal{U}_3 - \hat{u}\hat{u}^T \right) \bar{p}
 \nonumber \\
&= \frac{2 \pi^2 a_0}{k_0^2 Z_0} 
 \; \bar{p}^\dagger \left( \mathcal{U}_3 - \hat{u}\hat{u}^T \right) \bar{p}.
\end{align}
We can eliminate the normalization constant $a_0$ by making use of (\ref{eqn:blackbody_I}). In addition, we note that
\begin{align}
\bar{p}^\dagger \hat{u} = \hat{p} \cdot \hat{u} = 0 \\
\bar{p}^\dagger \, \mathcal{U}_3 \, \bar{p} = \bar{p}^\dagger \bar{p} = 1.
\end{align}
This gives us our final expression for the speckle power of a Lambertian scatterer:
\begin{align}
P_0 = \frac{2 \pi I}{k_0^2} = \frac{\lambda^2 I}{2 \pi}.
\end{align}
Here, $\lambda$ is the wavelength of light, and the intensity $I$ can be understood as the power emitted by the scattering surface per unit area of the surface. Meanwhile, the characteristic area $A_0$ of a speckle, which we defined in equation \ref{eqn:def_speckle_area} as the speckle power $P_0$ divided by the intensity $I$, is given by
\begin{align}
A_0 = \frac{2 \pi}{k_0^2} =  \frac{\lambda^2}{2 \pi}.
\end{align}

\bibliographystyle{ieeetr}
\bibliography{lidar_ref}

\begin{thebibliography}{10}

\bibitem{bbehroozpour_ieeecm2017}
B.~{Behroozpour}, P.~A.~M. {Sandborn}, M.~C. {Wu}, and B.~E. {Boser}, ``Lidar
  system architectures and circuits,'' {\em IEEE Communications Magazine},
  vol.~55, pp.~135--142, Oct 2017.

\bibitem{hdgriffiths_ecej1990}
H.~D. {Griffiths}, ``New ideas in {FM} radar,'' {\em Electronics Communication
  Engineering Journal}, vol.~2, pp.~185--194, Oct 1990.

\bibitem{cvpoulton_ol2017}
C.~V. Poulton, A.~Yaacobi, D.~B. Cole, M.~J. Byrd, M.~Raval, D.~Vermeulen, and
  M.~R. Watts, ``Coherent solid-state {LIDAR} with silicon photonic optical
  phased arrays,'' {\em Opt. Lett.}, vol.~42, pp.~4091--4094, Oct 2017.

\bibitem{samiller_cleo2018}
S.~A. Miller, C.~T. Phare, Y.-C. Chang, X.~Ji, O.~A.~J. Gordillo, A.~Mohanty,
  S.~P. Roberts, M.~C. Shin, B.~Stern, M.~Zadka, and M.~Lipson, ``512-element
  actively steered silicon phased array for low-power {LIDAR},'' in {\em
  Conference on Lasers and Electro-Optics}, p.~JTh5C.2, Optical Society of
  America, 2018.

\bibitem{cvpoulton_ieeeqe2019}
C.~V. {Poulton}, M.~J. {Byrd}, P.~{Russo}, E.~{Timurdogan}, M.~{Khandaker},
  D.~{Vermeulen}, and M.~R. {Watts}, ``Long-range {LiDAR} and free-space data
  communication with high-performance optical phased arrays,'' {\em IEEE
  Journal of Selected Topics in Quantum Electronics}, vol.~25, pp.~1--8, Sep.
  2019.

\bibitem{jriemensberger_nature2020}
J.~Riemensberger, A.~Lukashchuk, M.~Karpov, W.~Weng, E.~Lucas, J.~Liu, and
  T.~J. Kippenberg, ``Massively parallel coherent laser ranging using a soliton
  microcomb,'' {\em Nature}, vol.~581, pp.~164--170, May 2020.

\bibitem{blackmore_patent2019_phase}
tephen C.~Crouch and K.~Rupavatharam, ``Method and system for time separated
  quadrature detection of doppler effects in optical range measurements,''
  2019.

\bibitem{mjcollett_jmo1987}
M.~Collett, R.~Loudon, and C.~Gardiner, ``Quantum theory of optical homodyne
  and heterodyne detection,'' {\em Journal of Modern Optics}, vol.~34, no.~6-7,
  pp.~881--902, 1987.

\bibitem{marubin_ol2007}
M.~A. Rubin and S.~Kaushik, ``Squeezing the local oscillator does not improve
  signal-to-noise ratio in heterodyne laser radar,'' {\em Opt. Lett.}, vol.~32,
  pp.~1369--1371, Jun 2007.

\bibitem{kmnatarajan_sst2012}
C.~M. Natarajan, M.~G. Tanner, and R.~H. Hadfield, ``Superconducting nanowire
  single-photon detectors: physics and applications,'' {\em Superconductor
  Science and Technology}, vol.~25, p.~063001, apr 2012.

\bibitem{aesiegman_ao1966}
A.~E. Siegman, ``The antenna properties of optical heterodyne receivers,'' {\em
  Appl. Opt.}, vol.~5, pp.~1588--1594, Oct 1966.

\bibitem{jywang_ao1982}
J.~Y. Wang, ``Heterodyne laser radar-{SNR} from a diffuse target containing
  multiple glints,'' {\em Appl. Opt.}, vol.~21, pp.~464--476, Feb 1982.

\bibitem{prmcisaac_ieeetmtt1991}
P.~R. McIsaac, ``Mode orthogonality in reciprocal and nonreciprocal
  waveguides,'' {\em IEEE Transactions on Microwave Theory and Techniques},
  vol.~39, pp.~1808--1816, Nov 1991.

\bibitem{rfharrington_2001_reciprocity}
R.~F. Harrington, {\em Time-Harmonic Electromagnetic Fields}, pp.~116--120.
\newblock IEEE Press, 2001.

\bibitem{kblomstedt_pra2013}
K.~Blomstedt, T.~Set\"al\"a, J.~Tervo, J.~Turunen, and A.~T. Friberg, ``Partial
  polarization and electromagnetic spatial coherence of blackbody radiation
  emanating from an aperture,'' {\em Physical Review A}, vol.~88, p.~013824,
  Jul 2013.

\end{thebibliography}

\begin{appendices}

\clearpage
\section{Spectral density of windowed \\random signals}
\label{appendix:esd_windowed_random}
Suppose we have a stationary random signal $x(t)$. The autocorrelation $\hat{x}(\tau)$ is defined as
\begin{align}
\hat{x}(\tau) = \braket{x(t) x^\ast(t - \tau)},
\end{align}
and the power spectral density $S(f)$ is
\begin{align}
S(f) = \lim_{T \rightarrow \infty}
  \left<\frac{1}{T} \left| \int_0^T x(t) e^{-2 \pi j f t} dt \right|^2 \right>.
\end{align}
Here, $\braket{\cdot}$ denotes the expected value. From the Wiener-Khinchin theorem, the autocorrelation $\hat{x}(\tau)$ and power-spectral density $S(f)$ are Fourier transform pairs, related by
\begin{align}
S(f) = \int_{-\infty}^{\infty} \hat{x}(\tau) e^{-2 \pi j f \tau} d\tau.
\end{align}

Now, suppose we wish to take the Fourier transform of the random signal $x(t)$ when it is windowed by some window function $w(t)$. We define the windowed random signal $y(t)$ as
\begin{align}
y(t) = w(t) \, x(t)
\label{eqn:window_rand}
\end{align}
and denote its Fourier transform as $Y(f)$. For simplicity, we will assume that the random process has zero mean, i.e. $\braket{x(t)} = 0$. As a result, the expected value of $Y(f)$ will also be zero, since the Fourier transform and expection value operators commute with each other:
\begin{align}
\braket{Y(f)} 
  &= \left< 
       \int_{-\infty}^{\infty} y(t) e^{-2 \pi j f t} dt
    \right> \nonumber\\
  &= \int_{-\infty}^{\infty} w(t) \braket{x(t)} e^{-2 \pi j f t} dt \\
  &= 0.
\end{align}

The expected energy spectral density $\braket{|Y(f)|^2}$, however, does not have such a trivial result. By definition,
\begin{align}
\braket{|Y(f)|^2}
 &= \left< \left|
        \int_{-\infty}^{\infty} y(t) e^{-2 \pi j f t} dt
    \right|^2 \right> \nonumber \\
 &= \int_{-\infty}^{\infty} dt \int_{-\infty}^{\infty} dt'
        \left< y(t) y^\ast(t') \right> e^{-2 \pi j f (t - t')}.
\end{align}
Changing variables $\tau = t - t'$ in the inner integral yields
\begin{align}
\braket{|Y(f)|^2}
 &= \int_{-\infty}^{\infty} dt \int_{-\infty}^{\infty} d\tau
        \left< y(t) \, y^\ast(t - \tau) \right> e^{-2 \pi j f \tau}.
\label{eqn:window_rand_psd_0}
\end{align}
However, from our definition of $y(t)$ in equation \ref{eqn:window_rand}, we can expand $\left< y(t) \, y^\ast(t - \tau) \right>$ as
\begin{align}
\left< y(t)  y^\ast(t - \tau) \right> 
 &= w(t) \, w^\ast(t - \tau) \left< x(t)\, x^\ast(t - \tau) \right> \nonumber\\
 &= w(t) \, w^\ast(t - \tau) \, \hat{x}(\tau).
\end{align}
If we further change the order of integration, equation \ref{eqn:window_rand_psd_0} becomes
\begin{align}
\braket{|Y(f)|^2}
 &= \int_{-\infty}^{\infty} d\tau \, \hat{x}(\tau)  e^{-2 \pi j f \tau}
    \int_{-\infty}^{\infty} dt\, w(t) w^\ast(t - \tau).
\end{align}
The inner integral can be identified as the autocorrelation $\hat{w}(\tau)$ of the window function $w(t)$,
\begin{align}
\hat{w}(\tau) = \int_{-\infty}^{\infty} w(t) w^\ast(t - \tau),
\end{align}
allowing us to simplify our expression to
\begin{align}
\braket{|Y(f)|^2}
 &= \int_{-\infty}^{\infty} d\tau \, \hat{x}(\tau) \hat{w}(\tau) e^{-2 \pi j f \tau}.
\end{align}
Applying the convolution theorem, we find that the expected energy spectral density is
\begin{align}
\braket{|Y(f)|^2} = S(f) \ast \int_{-\infty}^{\infty} \hat{w}(\tau) e^{- 2 \pi j f \tau},
\label{eqn:window_rand_psd_1}
\end{align}
where $\ast$ is the convolution operator
\begin{align}
(f \ast g) (x) = \int_{-\infty}^{\infty} f(y) g(x - y) dy.
\end{align}
We can identify the integral in equation \ref{eqn:window_rand_psd_1} as the energy spectral density $|W(f)|^2$ of the window function $w(t)$, since
\begin{align}
|W(f)|^2 
&= \left| \int_{-\infty}^{\infty} w(t) e^{- 2 \pi j f t} dt \right|^2 \nonumber \\
&= \int_{-\infty}^{\infty} dt \int_{-\infty}^{\infty} dt'
     w(t) \, w(t') e^{-2 \pi j f (t - t')} \nonumber \\
&= \int_{-\infty}^{\infty} d\tau e^{-2 \pi j f \tau} 
     \int_{-\infty}^{\infty} dt \, w(t) \, w(t - \tau) \nonumber \\
&= \int_{-\infty}^{\infty} \hat{w}(\tau) e^{- 2 \pi j f \tau} d\tau.
\end{align}
Thus, the expected energy spectral density $\braket{|Y(f)|^2}$ is simply the power spectral density $S(f)$ of the random signal convolved with the energy spectral density $|W(f)|^2$ of the window function:
\begin{align}
\braket{|Y(f)|^2} = S(f) \ast |W(f)|^2.
\end{align}

\clearpage
\section{Modes of lossless waveguides}
\label{appendix:wg_modes}

\begin{theorem}
Suppose we have a waveguide that is simultaneously:
\begin{enumerate}
\item Reciprocal.
\item Lossless.
\item Has reflection symmetry across some plane perpendicular to $\hat{z}$.
\end{enumerate}
Then the transverse field components of any guided mode of this waveguide will be purely real.
\end{theorem}

\begin{proof}
To show this, we will build upon the results given in reference \cite{prmcisaac_ieeetmtt1991}. Consider a guided mode of the waveguide, which is any solution to Maxwell's equations of the form
\begin{align}
\mathbf{E}_m(x,y,z) &= \mathcal{E}_m(x,y) e^{j \beta z} \\
\mathbf{H}_m(x,y,z) &= \mathcal{H}_m(x,y) e^{j \beta z}
\end{align}
with a purely real propagation constant $\beta$. We denote the transverse components of the modal fields (i.e. the field components orthogonal to $z$) as $\mathcal{E}_{m,T}$ and $\mathcal{H}_{m,T}$.

From reference \cite{prmcisaac_ieeetmtt1991}, the following additional guided modes must exist:
\begin{enumerate}
\item Reflection symmetry implies that there also exists a mode with propagation constant $-i\beta$ and transverse modal fields $\mathcal{E}_{m,T}$ and $-\mathcal{H}_{m,T}$.
\item Losslessness implies that there also exists a mode with propagation constant $-i\beta$ and transverse modal fields $\mathcal{E}_{m,T}^\ast$ and $-\mathcal{H}_{m,T}^\ast$.
\end{enumerate}
If our waveguide does not have mode degeneracy (i.e. no two waveguide modes share the same propagation constant), we can conclude that these two backwards propagating modes must be the same since they have the same propagation constant $-j\beta$. This implies that the modal fields must be equal up to the some phase factor $\phi$:
\begin{align}
\mathcal{E}_{m,T} &= e^{j \phi} \mathcal{E}_{m,T}^\ast \\
-\mathcal{H}_{m,T} &= e^{j \phi} \mathcal{H}_{m,T}^\ast.
\end{align}
Rearranging yields
\begin{align}
e^{-j \phi / 2} \mathcal{E}_{m,T} &= \left( e^{-j \phi / 2} \mathcal{E}_{m,T} \right)^\ast \\
e^{-j \phi / 2} \mathcal{H}_{m,T} &= \left( e^{-j \phi / 2} \mathcal{H}_{m,T} \right)^\ast.
\end{align}
Thus, $e^{-j \phi / 2} \mathcal{E}_{m,T}$ and $e^{-j \phi / 2} \mathcal{H}_{m,T}$ are purely real. In other words, if we have a waveguide that is reciprocal, lossless, has reflection symmetry, and lacks mode degeneracy, the transverse modal fields $\mathcal{E}_m$ and $\mathcal{H}_m$ of any propagating modes can be taken to be real.

Waveguides with degenerate modes also satisfy this property, since they can be taken as the limiting case of distorted waveguides that lack mode degeneracy. For example, the fundamental mode of a circular core waveguide is degenerate due to rotational symmetry. However, a circular core waveguide can be understood as the limiting case of a sequence of elliptical core waveguides that gradually become more circular. Throughout this entire sequence, the fundamental mode will remain non-degenerate, and thus its transverse field components will remain purely real.
\end{proof}

\end{appendices}

\end{document}